\documentclass{article}

\usepackage[bibliosources=refs]{pomegranate}
\usepackage[shortlabels]{enumitem}

\DeclareDocumentTextCommand{\NC}{}{\textsf{NC}}
\DeclareDocumentTextCommand{\RNC}{}{\textsf{RNC}}
\DeclareDocumentTextCommand{\QuasiNC}{}{\textsf{quasi-NC}}
\DeclareDocumentTextCommand{\SharpP}{}{\textsf{\#P}}

\DeclareOperator{\sign}
\DeclareOperator{\PM}
\DeclareDocumentMathCommand{\crc}{}{\operatorname{circ}}
\DeclareDocumentMathCommand{\mis}{}{\operatorname{mismatch}}
\DeclareDocumentMathCommand{\F}{}{{\mathcal{F}}}
\DeclareDocumentMathCommand{\L}{}{{\mathcal{L}}}
\DeclareDocumentMathCommand{\W}{}{{\mathcal{W}}}
\DeclareDocumentMathCommand{\O}{}{{\mathcal{O}}}

\SetKwInOut{Input}{Input}
\SetKwInOut{Output}{Output}
\SetKwBlock{InParallel}{in parallel do}{end}
\SetKwFor{ParallelFor}{for}{in parallel do}{end}
\SetKwFunction{PerfectMatching}{\textsc{PerfectMatching}}
\SetKwFunction{PartialMatching}{\textsc{PartialMatching}}
\SetKwFunction{Support}{\textsc{Support}}
\SetKwFunction{Reduce}{\textsc{Reduce}}

\provideboolean{havetikz}
\setboolean{havetikz}{true}

\definecolor{Black}{RGB}{0, 0, 0}
\definecolor{LightGray}{RGB}{216, 216, 216}
\definecolor{Gray}{RGB}{127, 127, 127}
\definecolor{Orange}{RGB}{237, 125, 49}
\definecolor{LightOrange}{RGB}{251, 229, 214}
\definecolor{Yellow}{RGB}{255, 192, 0}
\definecolor{Blue}{RGB}{91, 155, 213}
\definecolor{LightBlue}{RGB}{222, 235, 247}
\definecolor{Green}{RGB}{112, 173, 71}
\definecolor{LightGreen}{RGB}{226, 240, 217}
\definecolor{Navy}{RGB}{68, 114, 196}
\definecolor{LightNavy}{RGB}{218, 227, 243}

\tikzstyle{node}=[circle, line width=1, draw=Black, fill=Gray, inner sep=3]
\tikzstyle{edge}=[line width=1, color=Gray]
\tikzstyle{highlight}=[line width=2, color=Orange]
\tikzstyle{set}=[line width=1, color=Black, dashed]
\tikzstyle{fillset}=[color=Black, fill=Gray, line width=1]
\makeatletter
\tikzset{use path/.code=\tikz@addmode{\pgfsyssoftpath@setcurrentpath#1}}
\makeatother

\tikzset{
	graph/.pic={
		\foreach \v/\coord in {{a/(-2, -1)}, {b/(-2, 1)}, {c/(-1, 0)}, {d/(2, -1)}, {e/(2, 1)}, {f/(1, 0)}} \node[node] (\v) at \coord {};
		\foreach \u/\v in {a/d, b/e, c/f, a/b, b/c, c/a, d/e, e/f, f/d} \draw[edge] (\u) -- (\v);
		\draw[set] (-3, 0) to[out=90, in=180] (-2, 1.5) to[out=0, in=90] (-0.5, 0) to[out=-90, in=0] (-2, -1.5) to[out=180, in=-90] cycle;
		\foreach \u/\v/\pos in {b/e/above, b/c/above, e/f/above, a/b/left, d/e/right} \draw (\u) edge[edge] node[pos=0.5, \pos, Black] {$\frac{1}{3}$} (\v);
		\foreach \u/\v/\pos in {a/d/below, c/f/above}
			\path (\u) edge[highlight] node[pos=0.5, \pos, Black] {$\frac{1}{3}\Orange{-\epsilon}$} (\v);
		\foreach \u/\v/\pos in {a/c/right, d/f/left}
			\path (\u) edge[highlight] node[pos=0.5, \pos, Black] {$\frac{1}{3}\Orange{+\epsilon}$} (\v);
		\node at (-3, 1.5) {$S$};
	}
}

\tikzset{
	contracted/.pic={
		\foreach \v/\coord in {{a/(-2, 0)}, {b/(-2, 0)}, {c/(-2, 0)}, {d/(2, -1)}, {e/(2, 1)}, {f/(1, 0)}} \node[node] (\v) at \coord {};
		\foreach \u/\v in {a/d, b/e, c/f, a/b, b/c, c/a, d/e, e/f, f/d} \draw[edge] (\u) -- (\v);
		\foreach \u/\v/\pos in {a/d/below, a/e/above, a/f/right, d/e/right, e/f/above, f/d/above}
			\path (\u) edge[edge] node[pos=0.5, Black, \pos] {$\frac{1}{3}$} (\v);
		\node at (-2.3, 0.3) {$S$};
	}
}

\tikzset{
	matching-polytope-base/.pic={
		\begin{scope}[scale=2]
			\coordinate (O) at (0.5, 0.5, 0.5);
			\coordinate (U1) at (1, 0.5, 0.5);
			\coordinate (U2) at (0.5, 1, 0.5);
			\coordinate (V) at (1, 1, 0.5);
			\coordinate (A1) at (0, 0, 0);
			\coordinate (A2) at (1, 0, 0);
			\coordinate (A3) at (1, 1, 0);
			\coordinate (A4) at (0, 1, 0);
			\coordinate (B1) at (0, 0, 1);
			\coordinate (B2) at (1, 0, 1);
			\coordinate (B3) at (1, 1, 1);
			\coordinate (B4) at (0, 1, 1);
			\coordinate (C1) at (1, 0.6, 0);
			\coordinate (C2) at (0.6, 1, 0);
			\coordinate (D1) at (1, 0.6, 1);
			\coordinate (D2) at (0.6, 1, 1);
			\coordinate (E1) at (-0.5, 1.5, -0.5);
			\coordinate (E2) at (-0.5, 1.5, 1.5);
			\coordinate (F2) at (-0.5, 1.5, 1);
			\coordinate (E4) at (1.5, -0.5, -0.5);
			\coordinate (E3) at (1.5, -0.5, 1.5);
			\coordinate (F3) at (1.5, -0.5, 1);
			\coordinate (M1) at (0, 0, 0.5);
			\coordinate (M2) at (-0.5, -0.5, 0.5);
		\end{scope}
	}
}

\tikzset{
	bipartite-matching-polytope0/.pic={
		\pic{matching-polytope-base};
		\begin{scope}[fill=Gray, opacity=0.3]
			\fill (A1) -- (A2) -- (A3) -- (A4) -- cycle;
			\fill (A1) -- (A2) -- (B2) -- (B1) -- cycle;
			\fill (A1) -- (A4) -- (B4) -- (B1) -- cycle;
		\end{scope}
	}
}

\tikzset{
	bipartite-matching-polytope1/.pic={
		\begin{scope}[line width=1, Black, dashed, opacity=0.6]
			\foreach \u/\v in {A1/A4, A1/A2, A1/B1}
				\draw (\u) -- (\v);
		\end{scope}
	}
}
\tikzset{
	bipartite-matching-polytope2/.pic={
		\begin{scope}[fill=Gray, opacity=.3]
			\fill (B1) -- (B2) -- (B3) -- (B4) -- cycle;
			\fill (B3) -- (A3) -- (A4) -- (B4) -- cycle;
			\fill (B3) -- (A3) -- (A2) -- (B2) -- cycle;
		\end{scope}
		\begin{scope}[line width=1, Black, opacity=1.0]
			\foreach \u/\v in {A2/A3, A3/A4, B1/B2, B2/A2, B1/B4, B4/A4, B3/B2, B3/B4, B3/A3}
				\draw (\u) -- (\v);
		\end{scope}
	}
}

\tikzset{
	bipartite-matching-polytope/.pic={
		\pic{bipartite-matching-polytope0};
		\pic{bipartite-matching-polytope1};
		\pic{bipartite-matching-polytope2};
	}
}

\tikzset{
	bipartite-matching-polytope01/.pic={
		\pic{bipartite-matching-polytope0};
		\pic{bipartite-matching-polytope1};
	}
}

\tikzset{
	matching-polytope0/.pic={
		\pic{matching-polytope-base};
		\begin{scope}[fill=Gray, opacity=0.3]
			\fill (A1) -- (A2) -- (C1) -- (C2) -- (A4) -- cycle;
			\fill (A1) -- (A2) -- (B2) -- (B1) -- cycle;
			\fill (A1) -- (A4) -- (B4) -- (B1) -- cycle;
		\end{scope}
	}
}

\tikzset{
	matching-polytope1/.pic={
		\begin{scope}[line width=1, Black, dashed, opacity=0.6]
			\foreach \u/\v in {A1/A4, A1/A2, A1/B1}
				\draw (\u) -- (\v);
		\end{scope}
	}
}

\tikzset{
	matching-polytope2/.pic={
		\begin{scope}[fill=Gray, opacity=.3]
			\fill (B1) -- (B2) -- (D1) -- (D2) -- (B4) -- cycle;
			\fill (B2) -- (D1) -- (C1) -- (A2) -- cycle;
			\fill (B4) -- (D2) -- (C2) -- (A4) -- cycle;
			\fill (C1) -- (D1) -- (D2) -- (C2) -- cycle;
		\end{scope}
		\begin{scope}[line width=1, Black, opacity=1.0]
			\foreach \u/\v in {B1/B2, B2/A2, B1/B4, B4/A4, A2/C1, C1/C2, C2/A4, B2/D1, D1/D2, D2/B4, C1/D1, C2/D2}
				\draw (\u) -- (\v);
		\end{scope}
	}
}

\tikzset{
	matching-polytope/.pic={
		\pic{matching-polytope0};
		\pic{matching-polytope1};
		\pic{matching-polytope2};
	}
}

\tikzset{
	matching-polytope01/.pic={
		\pic{matching-polytope0};
		\pic{matching-polytope1};
	}
}

\tikzset{
	bipartite-matching-polytope-seq/.pic={
		\pic at (-2.5, 0) {bipartite-matching-polytope01};
		\node at (O) {\LARGE $\Orange{\bullet}$};
		\draw[Yellow, line width=2, ->] (O) -- (U1);
		\pic at (-2.5, 0) {bipartite-matching-polytope2};
		\node at (0, 0.5) {\Huge $+$};
		\pic at (1.5, 0) {bipartite-matching-polytope01};
		\node at (O) {\LARGE $\Orange{\bullet}$};
		\draw[Yellow, line width=2, ->] (O) -- (U2);
		\pic at (1.5, 0) {bipartite-matching-polytope2};
		\node at (0, -1) {$\underbrace{\hspace{15em}}$};
		\node at (0, -1.6) {\Huge $\Downarrow$};
		\pic at (-0.5, -4.2) {bipartite-matching-polytope01};
		\node at (O) {\LARGE $\Orange{\bullet}$};
		\draw[Yellow, line width=2] (O) -- (barycentric cs:C1=1,D1=1,C2=1,D2=1);
		\draw[Yellow, line width=2, ->] (barycentric cs:C1=1,D1=1,C2=1,D2=1) -- (V);
		\pic at (-0.5, -4.2) {bipartite-matching-polytope2};
	}
}

\tikzset{
	matching-polytope-seq/.pic={
		\pic at (-2.5, 0) {matching-polytope01};
		\node at (O) {\LARGE $\Orange{\bullet}$};
		\draw[Yellow, line width=2, ->] (O) -- (U1);
		\pic at (-2.5, 0) {matching-polytope2};
		\node at (0, 0.5) {\Huge $+$};
		\pic at (1.5, 0) {matching-polytope01};
		\node at (O) {\LARGE $\Orange{\bullet}$};
		\draw[Yellow, line width=2, ->] (O) -- (U2);
		\pic at (1.5, 0) {matching-polytope2};
		\node at (0, -1) {$\underbrace{\hspace{15em}}$};
		\node at (0, -1.6) {\Huge $\Downarrow$};
		\pic at (-0.5, -4.2) {matching-polytope01};
		\node at (O) {\LARGE $\Orange{\bullet}$};
		\draw[Yellow, line width=2] (O) -- (barycentric cs:C1=1,D1=1,C2=1,D2=1);
		\pic at (-0.5, -4.2) {matching-polytope2};
		\draw[Yellow, line width=2, ->] (barycentric cs:C1=1,D1=1,C2=1,D2=1) -- (V);
	}
}

\tikzset{
	hyperplane/.pic={
		\pic{matching-polytope-base};
		\fill[Blue, opacity=0.9] (B4) -- (A4) -- (A2) -- (B2) -- (F3) -- (E4) -- (E1) -- (F2) -- cycle;
		\draw[Orange, line width=2, ->] (M1) -- (M2) node[Black, pos=1.2] {$w$};
		
		\pic{matching-polytope0};
		\pic{matching-polytope1};
		\draw[Orange, line width=2] (O) -- (M1);
		\fill[Blue, opacity=0.9] (A4) -- (B4) -- (B2) -- (A2) -- cycle;
		\node at (O) {\LARGE $\Orange{\bullet}$};
		\draw[Yellow, line width=2, ->] (O) -- (U1);
		\draw[Yellow, line width=2, ->] (O) -- (U2);
		\pic{matching-polytope2};
		\fill[Blue, opacity=0.9] (E2) -- (F2) -- (F3) -- (E3) -- cycle;
		\node[rotate=45] at (barycentric cs:E1=7,E2=7,E3=1,E4=1) {$\dotp{w, x}=\text{const}$};
		\node (minimizer) at (barycentric cs:C1=1,C2=1,D1=1,D2=1) {\LARGE $\bullet$};
		\node (text) at (barycentric cs:O=1,V=-2) {minimizer of $\dotp{w,x}$};
		\draw[dashed, ->] (text) edge[out=-120, in=30] (minimizer);
	}
}

\tikzset{
	shrinking1/.pic={
		\foreach \u/\coord in {{a/(0, 0)}, {b/(2, 0)}, {c/(1, 0.5)}, {d/(0, -2)}, {e/(-0.5, -1)}, {f/(-1, -2)}, {g/(1, -2)}, {h/(2, -1)}, {i/(1, -1)}, {j/(2.3, -2.3)}}
			\node[node] (\u) at \coord {\u};
		\foreach \u/\v in {a/b, b/c, c/d, d/e, e/f, f/d, d/g, g/h, h/i, i/a, a/e, h/b, i/g, g/j, h/j}
			\draw[edge] (\u) -- (\v);
		\draw[set] (-1.5, -2) to[out=90, in=-180] (-0.5, -0.5) to[out=0, in=90] (0.5, -2) to[out=-90, in=-90] cycle;
	}
}

\tikzset{
	shrinking2/.pic={
		\foreach \u/\coord in {{a/(0, 0)}, {b/(2, 0)}, {c/(1, 0.5)}, {def/(-0.8, -1.3)}, {g/(1, -2)}, {h/(2, -1)}, {i/(1, -1)}, {j/(2.3, -2.3)}}
			\node[node] (\u) at \coord {\u};
		\foreach \u/\v in {a/b, b/c, c/def, def/g, g/h, h/i, i/a, a/def, h/b, i/g, g/j, h/j}
			\draw[edge] (\u) -- (\v);
		\draw[set] (0.7, -0.7) to[out=45, in=60] (2.5, -1) to[out=-120, in=30] (1, -2.5) to[out=-150, in=-135] cycle;
		\draw[set] (-1.2, -2.3) to[out=-145, in=-135] (-0.4, 0.5) to[out=45, in=50] (1.6, 0.5) to[out=-130, in=35] cycle;
	}
}

\tikzset{
	shrinking3/.pic={
		\foreach \u/\coord in {{acdef/(0, -0.5)}}
			\node[node, color=Orange, fill=LightOrange] (\u) at \coord {\u};
		\foreach \u/\coord in {{b/(2, 0)},  {ghi/(1.2, -1.2)}, {j/(2.3, -2.3)}}
			\node[node] (\u) at \coord {\u};
		\foreach \u/\v in {acdef/b, acdef/ghi, ghi/b, ghi/j}
			\draw[edge] (\u) -- (\v);
	}
}

\tikzset{
	shrinking/.pic={
		\pic{shrinking1};
		\node at (3, -1) {\Huge $\Rightarrow$};
		\pic at (5, 0) {shrinking2};
		\node at (8.2, -1) {\Huge $\Rightarrow$};
		\pic at (9.5, 0) {shrinking3};
	}
}

\tikzset{
	sameblock/.pic={
		\node[node] (a1) at (-1, -1) {};
		\node[node] (a2) at (1, -1) {};
		\node[node] (b1) at (-1, 1) {};
		\node[node] (b2) at (1, 1) {};
		\coordinate (a0) at (-1, -2); 
		\coordinate (a3) at (1, -2);
		\coordinate (b0) at (-1, 2);
		\coordinate (b3) at (1, 2);
		\draw (a0) edge[highlight] node[pos=0.5, Black, right] {$-\epsilon_1$} (a1);
		\path (a1) edge[highlight] node[pos=0.5, Black, above] {$+\epsilon_1$} (a2);
		\path (a2) edge[highlight] node[pos=0.5, Black, left] {$-\epsilon_1$} (a3);
		\path (a3) edge[dashed, out=-120, in=-60] (a0);
		\path (b0) edge[highlight] node[pos=0.5, Black, right] {$-\epsilon_2$} (b1);
		\path (b1) edge[highlight] node[pos=0.5, Black, below] {$+\epsilon_2$} (b2);
		\path (b2) edge[highlight] node[pos=0.5, Black, left] {$-\epsilon_2$} (b3);
		\path (b3) edge[dashed, out=120, in=60] (b0);
		\draw[set] (-2, 0) to[out=90, in=180] (0, 2) to[out=0, in=90] (2, 0) to[out=-90, in=0] (0, -2) to[out=180, in=-90] cycle;
		\node at (0, 0) {$S$};
	}
}

\title{Matching is as Easy as the Decision Problem,\\ in the NC Model}

\author{Nima Anari}
\affil{Stanford University, \textsf{anari@cs.stanford.edu}}
\author{Vijay V.~Vazirani}
\affil{University of California, Irvine, \textsf{vazirani@ics.uci.edu}}

\begin{document}
	\maketitle
	
	\begin{abstract}
		Is matching in \NC{}, i.e., is there a deterministic fast parallel algorithm for it? This has been an outstanding open question in TCS for over three decades, ever since the discovery of randomized \NC{} matching algorithms \cite{KUW85, MVV87}.  Over the last five years, the theoretical computer science community has launched a relentless attack on this question, leading to the discovery of several powerful ideas. We give what appears to be the culmination of this line of work: An \NC{} algorithm for finding a minimum-weight perfect matching in a general graph with polynomially bounded edge weights, provided it is given an oracle for the decision problem. Consequently, for settling the main open problem, it suffices to obtain an \NC{} algorithm for the decision problem. We believe this new fact has qualitatively changed the nature of this open problem.

		All known efficient matching algorithms for general graphs follow one of two approaches: given by \textcite{Edmonds65} and \textcite{Lovasz79}. Our oracle-based algorithm follows a new approach and uses many of the ideas discovered in the last five years.
		
		The difficulty of obtaining an \NC{} perfect matching algorithm led researchers to study matching vis-a-vis clever relaxations of the class \NC{}. In this vein, recently \textcite{GG15} gave a pseudo-deterministic \RNC{} algorithm for finding a perfect matching in a bipartite graph, i.e., an \RNC{} algorithm with the additional requirement that on the same graph, it should return the same (i.e., unique) perfect matching for almost all choices of random bits. A corollary of our reduction is an analogous algorithm for general graphs. 
	\end{abstract}

	\section{Introduction}\label{sec:intro}

Is matching in \NC, i.e., is there a deterministic fast parallel\footnote{That runs in polylogarithmic time using polynomially many processors.} algorithm for finding a perfect or, more generally, a maximum matching in a general graph? This has been an outstanding open question in theoretical computer science for over three decades, ever since the discovery of \RNC{} matching algorithms \cite{KUW85,MVV87}. Over the last five years, the TCS community has launched a relentless attack on this question, leading to the discovery of numerous powerful ideas \cite{FGT16, ST17, GG15, AV18, Sankowski18}. We give what appears to be the culmination\footnote{We note that since the appearance of this paper in aXiv, in January 2019, we are not aware of any new development on the "Is matching in \NC?" question, in contrast to the frenzied activity over the prior years.} of this line of work: An \NC{} algorithm for finding a minimum weight perfect matching in a general graph with polynomially bounded edge weights, provided it is given an oracle, say $\O$, for the decision problem. Consequently, for settling the main open problem, it suffices to obtain an \NC{} algorithm for the decision problem. We believe this new fact has qualitatively changed the nature of this open problem. Henceforth, by \emph{small weights} we will mean \emph{polynomially bounded edge weights} and \emph{acronym MWPM} will be short for \emph{minimum weight perfect matching.}

 The difficulty of obtaining an \NC{} matching algorithm led researchers to study matching vis-a-vis certain clever relaxations of the class \NC{}. One such relaxation is pseudo-deterministic \RNC{}. This is an \RNC{} algorithm with the additional property that on the same graph, it must return the same (i.e., unique) solution for almost all choices of random bits \cite{GG11, GG15}. Recently, \Textcite{GG15} gave such an algorithm for perfect matching in bipartite graphs. A second relaxation of \NC{} is \QuasiNC{}, under which the algorithm must run in polylogarithmic time, though it can use $O(n^{\log^{O(1)} n})$ processors; see \cref{sec.history} for results obtained for this model.
 
 A corollary of our result extends \cite{GG15} to general graphs as follows: The precise decision problem for our result is: Given a graph $G$ with small weights and a number $W$, is there a perfect matching of weight at most $W$ in $G$. Since binary search over $W$ will take $O(\log n)$ iterations, this is \NC{} equivalent to: Find the weight of a minimum weight perfect matching in $G$. This question is easy to answer in \RNC{} with inverse-polynomial probability of error using the algorithm\footnote{Note that the \RNC{} algorithm of \cite{MVV87} also finds a minimum weight perfect matching with high probability; however, unlike the weight, the latter is not guaranteed to be the same with any sizable probability.} of \cite{MVV87}. Therefore, using this \RNC{} algorithm in place of the oracle, we get an \RNC{} matching algorithm with the property that in a run, all queries to the decision problem will be answered correctly with overwhelming probability. Whenever the latter happens, the algorithm outputs the same (unique) perfect matching. Hence this is a pseudo-deterministic \RNC{} matching algorithm.  
 
 All known efficient matching algorithms for general graphs follow one of two approaches: given by \textcite{Edmonds65} and \textcite{Lovasz79}. Our oracle-based algorithm follows a new approach and uses many of ideas discovered in the last five years. The contributions of various authors is given in detail in Section \ref{sec:overview}; here we mention two main ingredients. Our algorithm uses the overall structure, as well as an \NC{} algorithm for finding a balanced viable set (see Section \ref{sec.balanced}), from the recent \NC{} algorithm of \cite{AV18} for finding a perfect matching in planar graphs. (Since oracle $\O$ can be implemented in \NC{} for planar graphs, our current paper yields a simpler \NC{} algorithm for finding a perfect matching in planar graphs.) The second key ingredient is an \NC{} algorithm for finding a maximal laminar family of tight odd sets in a given face of the perfect matching polytope. This follows from the works of \textcite{CGS15} and \textcite{Sankowski18}.

\begin{notation}
\label{not.oracle}
$\O$ will represent the oracle that answers, in one step, the decision question: 
Given a graph $G$ with small weights and a number $W$, is there a perfect matching of weight at most $W$ in $G$?	
\end{notation}

Our main result is:

\begin{theorem}\label{thm:main}
	There is an \NC{} algorithm for finding a MWPM in general graphs with small weights, provided the algorithm is given access to oracle $\O$ for the decision problem. The latter is: Given a graph $G$ with small weights and a target weight $W$, is there a perfect matching of weight at most $W$ in $G$?
\end{theorem}

\begin{corollary}
	There is an \NC{} algorithm for finding a maximum matching in general graphs, provided the algorithm is given access to oracle $\O$.
\end{corollary}

\begin{corollary}
\label{cor:pm}
	There is a pseudo-deterministic \RNC{} algorithm for finding a minimum weight perfect matching in general graphs with small weights.
\end{corollary}

We further show that our algorithm only need to call the decision oracle for minors of the input graph.

\begin{theorem}\label{thm:minor}
	Let $\F$ be a minor-closed family of graphs. If there is an \NC{} algorithm for deciding whether a perfect matching of weight at most $W$ exists in graphs from $\F$, weighted with polynomially small weights, then there is also an \NC{} algorithm for \emph{finding} a MWPM in such graphs.
\end{theorem}

\subsection{Related work and a brief history of parallel matching algorithms}
\label{sec.history}

The recent surge in activity on this problem was initiated by the elegant work of \textcite{FGT16} giving a \QuasiNC{} algorithm for perfect matching in bipartite graphs. The essential idea underlying their algorithm is to give a partial derandomization of the Isolation Lemma. In the process, they introduced some powerful ideas which were crucially used in later works and are detailed in Section \ref{sec:overview}. This was followed by the \QuasiNC{} algorithm of \textcite{ST17} for non-bipartite graphs. This work clarified the basic difficulty encountered in such graphs and ways of dealing with them; see Section \ref{sec:overview} for details.

The very first result on obtaining parallel matching algorithms was that the decision problem, of determining if a graph has a perfect matching, can be solved in \RNC. This is a folklore result -- it follows in a straightforward manner from Lovasz's \cite{Lovasz79} matching algorithm and Csanky's result \cite{Csanky} that the determinant of a matrix can be computed in \NC. 

The first \RNC{} algorithm for the search problem, of actually finding a perfect matching, was obtained by \textcite{KUW86}. This was followed by a simpler and more versatile algorithm due to \textcite{MVV87};  besides perfect matching, it also yielded \RNC{} algorithms for the problem of exact matching (see \cref{sec:discussion}) and for finding a MWPM in a graph with small weights. The latter fact is crucially used for obtaining pseudo-deterministic \RNC{} algorithms for bipartite graphs \cite{GG15} and general graphs (current paper). The ``philosophy'' behind \cite{MVV87} will be useful for dealing with a difficulty that arises in the design of the current algorithm as well, so it is recalled below\footnote{Under the \NC{} model, any one processor does not even have enough time to read the entire input, and hence can perform only local computations. On the other hand, a perfect matching is a global object, unlike say, a maximal independent set. Further difficulties arise from the fact that the number of perfect matchings in a graph can vary widely, all the way from one to exponentially many (assuming it has at least one). If there were a unique perfect matching in the graph, the algorithm's task would become a lot simpler. \Cite{MVV87} achieve uniqueness via their probabilistic fact, the Isolating Lemma: under an assignment of randomly chosen small weights to the edges it claims that the MWPM will be unique with high probability.}.

The matching problem occupies an especially distinguished position in the theory of algorithms: Some of the most central notions and powerful tools within this theory were discovered in the context of an algorithmic study of this problem, including the notion of polynomial time solvability \cite{Edmonds65}, the counting class \SharpP{} \cite{Valiant79} and a polynomial time equivalence between random generation and approximate counting for self-reducible problems \cite{JVV}, which lies at the core of the Markov chain Monte Carlo method. The perspective of parallel algorithms has also led to such a gain, namely the Isolation Lemma \cite{MVV87}, which has found several applications in complexity theory and algorithms; see the Wikipedia page \cite{IsolationL}. Due to the fundamental insights gained from an algorithmic study of matching, and the possibility of additional insights, the problem of obtaining an \NC{} algorithm for it has remained a premier open question ever since the 1980s.

The first substantial progress on this question was made for the case of planar bipartite graphs by \textcite{MN89} via a flow-based approach, followed by \textcite{MV00} using the fact that there is an \NC{} algorithm for counting perfect matchings in planar graphs. The long-standing problem of extending this result to non-bipartite planar graphs was resolved by \textcite{AV18}. Subsequently, \textcite{Sankowski18} also got the same result using different ideas. \Cite{AV18} also extended their algorithm to constant genus graphs. Subsequently, \textcite{EV18} gave an \NC{} algorithm for perfect matching in one-crossing-minor-free graphs, which include $K_5$-free graphs and $K_{3,3}$-free graphs; the resolution of the latter class settles an open problem asked in \cite{Vazirani89}.

The notion of a pseudo-deterministic algorithm with polynomial expected running time was given by \textcite{GG11} and was applied to several number theoretic and cryptographic problems. The notion of pseudo-deterministic \RNC{} algorithms was defined by \textcite{GG15}.

Recently, algorithms were also obtained for the generalization of bipartite matching to the linear matroid intersection problem by \textcite{GT17}, and to a further generalization of finding a vertex of a polytope with faces given by totally unimodular constraints, by \textcite{GTV17}.

\subsection{What is the ``right'' decision problem?}
\label{sec:right}

Consider the following two decision problems for perfect matching:
\begin{enumerate} 
	\item Given a graph $G$ with small weights and a target $W$, is there a perfect matching of weight at most $W$ in $G$?
	\item Does graph $G$ have a perfect matching?
\end{enumerate}

Clearly, the second can be reduced to the first and is therefore ``easier''. This leads to a legitimate question: why not attempt to reduce, in \NC, the search problem to the second decision problem? Our experience suggests that the first problem is much more basic for the setting at hand. We next provide evidence to this effect. 

Seeking a MWPM in a graph with small weights was the central problem in the work of \cite{MVV87}. The Isolating Lemma helped find small weights under which there was a unique MWPM. The second half of \cite{MVV87} gave an \NC{} algorithm for finding this (unique) perfect matching, using the Tutte matrix of the graph and matrix inversion; the latter is known to be in \NC{} \cite{Csanky}. Ever since then, perhaps the most used avenue for obtaining an \NC{} matching algorithm was to derandomize the Isolating Lemma. This would deterministically yield small weights under which there is a unique MWPM, and it could be found using the second half of \cite{MVV87}. 

The question of MWPM in a graph with small weights plays a central role in \NC-type approaches to all non-bipartite, and even some bipartite, perfect matching algorithms: partial derandomization leading to \QuasiNC{} algorithms \cite{FGT16,ST17}, resolution of the open problem of non-bipartite planar graphs \cite{AV18,Sankowski18}, and quasi-deterministic \RNC{} algorithms for bipartite \cite{GG15} and general graphs (current paper).  

In mathematics, sometimes solving a more general problem turns out to be easier than solving the special case, if the former has a better ``behavior''. Our belief is that this is the case here. The main avenue studied for solving the second decision problem was by derandomizing polynomial identity testing \cite{??}. However, more than three decades of work on the latter has yielded no substantial results. We believe it is time to wholeheartedly attack the first decision problem. Going forward, that is the main message of our paper.

\subsection{Bipartite vs non-bipartite matching: An intriguing phenomenon}
\label{sec:curious}

Decades of algorithmic work on the matching problem, from numerous perspectives, exhibits the following intriguing phenomenon: The bipartite case gets solved first. Then, using much more elaborate machinery, involving structural facts and algorithmic insights, the general graph case  follows, yielding the exact same result! This phenomenon is made all the more fascinating by the fact that the ``elaborate machinery'' consists not of one fact but numerous different structural properties and mathematical facts which happen to be just right for the problem at hand! We give a number of examples below.

The duality between maximum matching and minimum vertex cover for bipartite graphs extends to general graphs via the notion of an odd set cover, see \textcite{LP09}. The formulation of the perfect matching polytope for bipartite graphs extends by introducing constraints corresponding to odd sets \cite{Edmonds65}. Polynomial time algorithms for maximum matching and maximum weight matching in bipartite graphs generalize via the notion of blossoms \cite{LP09}. The most efficient known algorithm for maximum matching in bipartite graphs \cite{HK73, Karzanov73} obtained via an alternating breadth first search, extends via a much more elaborate algorithm with the same running time using the graph search procedure of double depth first search \cite{MV80} and blossoms defined from the perspective of minimum length alternating paths \cite{Vazirani94}. The \RNC{} matching algorithms \cite{KUW86, MVV87} use Tutte's theorem to extend to general graphs. The randomized matching algorithm of \cite{RV89} uses Tutte's theorem and a theorem of Frobenius about ranks of sub-matrices of skew-symmetric matrices. 

More recent work exhibits this phenomenon as well. The \QuasiNC{} algorithm of \textcite{FGT16} for bipartite graphs extends by handling tight odd cuts appropriately \cite{ST17}. The \NC{} algorithm of \cite{MV00} for planar bipartite graphs was extended to non-bipartite graphs via Edmonds' formulation of the perfect matching polytope \cite{Edmonds65}, an \NC{} algorithm for max-flow in planar graphs \cite{Johnson87}, and a result of \textcite{PR82} proving that the Gomory-Hu tree of a graph must contain a tight odd cut, and an elaborate \NC{} algorithm for uncrossing tight odd cuts \cite{AV18}. In the same vein, the current paper is extending the pseudo-deterministic \RNC{} bipartite algorithm of \cite{GG15} by giving a way of dealing with tight odd cuts in Edmonds' formulation of the perfect matching polytope \cite{Edmonds65} and using an \NC{} procedure for finding a maximal laminar family of tight odd cuts \cite{CGS15,Sankowski18}.

	\section{Overview and Technical Ideas}\label{sec:overview}

Most of this paper will concentrate on the problem of finding a perfect matching in a general graph in \NC{}, given oracle $\O$. In \cref{sec:weighted} we will extend our ideas to finding a MWPM for small weights; an algorithm for finding a maximum matching in a general graph in \NC{} will easily follow. In this section, we will also give a number of key definitions which will be used throughout the  paper.

\subsection{The bipartite case}

For ease of comprehension, we will first give an outline of a proof of \cref{thm:main} for the case of bipartite graphs. Such a proof can be gleaned from the paper of \textcite{GG15}; however, to the best of our knowledge, this important fact was not derived so far. Below, we build on the \QuasiNC{} algorithm of \textcite{FGT16} to obtain a somewhat simpler proof of this result.

The algorithm of \cite{FGT16} first finds a point in the interior of the perfect matching polytope and then iteratively moves to lower dimensional faces of this polytope, terminating when a vertex of the polytope is reached; this will be a perfect matching. 

\begin{definition}
In a general graph $G = (V, E)$ with edge weight function $w$, an edge $e$ is called an \emph{allowed edge} if it participates in MWPM. Let $E[w]$ denote the set of all allowed edges. Edges in the complement of this set will be called \emph{disallowed edges}.
\end{definition}

Assume $w$ are small weights and let $\PM[w]$ denote the face of the polytope containing all fractional and integral MWPMs w.r.t.\ $w$. Since we are in the bipartite case, $\PM[w]$ has a simple description: It is defined by the set of disallowed edges, since they are set to zero, or equivalently its complement, i.e., the set of allowed edges, $E[w]$. The description of the algorithm given above can be refined to: Iteratively modify the weight vector $w$ so that the dimension of face $\PM[w]$ keeps dropping, and equivalently $E[w]$ keeps getting sparser, until $E[w]$ is a perfect matching.

As argued earlier, using oracle $\O$, we can find the weight of a MWPM in $G$. Further, it is easy to see that for a given edge $e$, we can  determine in \NC{} if $e$ participates in a MWPM, i.e., if $e \in E[w]$. Repeating for all edges in parallel, we get the following easy fact for general graphs as well:

\begin{fact}
\label{fact.allowed}	
Given a graph $G = (V, E)$ and small weights $w$, and given oracle $\O$, we can compute $E[w]$ in \NC.
\end{fact}

The following is a fundamental notion in all recent \NC-type matching algorithms: 

\begin{definition}
\label{def:circulation}
(\textcite{Datta10})
Given a general graph $G$ with edge-weights $w$ and an even cycle $C$ in it, number the edges of $C$ consecutively, starting from an arbitrary edge. Then the \emph{circulation of cycle $C$} is the absolute value of the difference of the sum of weights of odd-numbered and even-numbered edges and is denoted by $\crc_w(C)$.
\end{definition}

 It is easy to prove that if the MWPM in $G$ is not unique, then any cycle in the symmetric difference of two such matchings must have zero circulation. It follows that if we find a weight vector $w$ such that each cycle in $G$ has nonzero circulation, then the MWPM must be unique and can be found in \NC{}. The next fact shows how to achieve this one cycle at a time. 
  
  \begin{fact}
  \label{fact:bipartite}
  (\textcite{FGT16})
  In a bipartite graph, let cycle $C \subseteq E[w]$ have $\crc_w(C)=0$. Let $G$ denote the graph on edge set $E[w]$. Assign small weights $w'$ to edges $E[w]$ so that $\crc_w'(C) > 0$. Then $C$ will not be present in $E[w']$, i.e., at least one of its edges will be dropped in going from $E[w]$ to $E[w']$. We will say that $C$ got \emph{destroyed}.
    \end{fact}

  Hence, if we find a weight vector that destroys all cycles of $G$, we would be done. However, $G$ may have exponentially many cycles, so this is non-trivial. One of the key ideas of \cite{FGT16} is a systematic way of destroying cycles: They iteratively destroy cycles of length $4, 8, 16, \dots, n$; clearly, the number of iterations needed is $O(\log n)$. In the first round, $G$ has at most $O(n^4)$ cycles of length 4. \Textcite{FGT16} show that if all cycles of length at most $2^i$ have already been destroyed, then there are at most $O(n^4)$ cycles of length at most $2^{i+1}$ left. Hence, in each iteration only $O(n^4)$ cycles need to be destroyed.

Suppose the current iteration starts with small weights $w$ under which all cycles of length at most $2^i$ have already been destroyed. In this iteration, the algorithm finds a weight vector $w'$ for the edges in $E[w]$ under which all cycles of length at most $2^{i+1}$ are destroyed. The following fact will play a central role in the current paper as well:

\begin{fact}
\label{fact:functions}
(\textcite{FGT16})
In order to destroy any set of $s$ cycles, it suffices to try certain well-chosen $O(n^2s)$ integral weight vectors each of which uses numbers that are $O(n^2s)$; one of these vectors is sure to work. 
 \end{fact}
 
Since in the current iteration $s = O(n^4)$, at most $O(n^6)$ weight vectors suffice. The algorithm for choosing a weight vector that works is as follows. In parallel, for each of the $O(n^6)$ weight vectors, $y$, compute $E[y]$ and find the girth of the resulting graph; this can easily be done in \NC. Pick the lexicographically first weight vector, say $w'$, such that $E[w']$ has girth $>2^{i+1}$. Clearly, $w'$ destroys all cycles of length at most $2^{i+1}$.

\subsection{Extension to general graphs}
\label{sec.extension}

In a wide range of computational models, matching algorithms for general graphs are far harder than for bipartite graphs, mainly because they need to handle odd cycles in special ways. The set of constraints capturing the perfect matching polytope is also more complex: it includes exponentially many odd set constraints. An odd set $S \subset V$ which satisfies this constraint with equality is called a \emph{tight odd set}. The description of face $\PM[w]$ is also much more involved: in addition to edges $E[w]$, we need a maximal laminar family of tight odd sets, say $\L$; see \cref{sec:face}. 

\paragraph{The ``engine'' underlying our algorithm:} Analogous to  \Cref{fact:bipartite}, which yielded the ``engine'' for the bipartite case, there is an ``engine'' underlying our algorithm as well -- it  iteratively reduces the size of the graph. This engine can be thought of as composed of three components which draw on different domains to establish structural facts and algorithms.

\subsubsection{Component based on the structure of the perfect matching polytope} 
We first note that 
 \Cref{fact:bipartite} does not hold in general graphs:  a non-bipartite graph may have an even cycle $C \subseteq E[w]$ with $\crc_w(C) > 0$. The reason is the presence of a tight odd set. As a result, \cref{fact:bipartite} needs to be enhanced to the fact stated below. We will say that a cycle $C$ \emph{crosses} a tight odd set $S$ if $C$ has vertices in $S$ as well as in $(V-S)$. Similarly, edge $e$ \emph{crosses $S$} if one of its endpoints is in $S$ and the other is in $V-S$.

\begin{fact}
\label{fact:cross}
(\textcite{ST17})
	In a general graph $G$, suppose even cycle $C \subseteq E[w]$ has $\crc_w(C) > 0$. Then, there must be a tight odd set $S$ such that $C$ \emph{crosses} $S$.
\end{fact}

This is illustrated in \cref{fig:blocked}. In this graph, the three edges in $\delta(S)$ have weight 1 and the rest have weight 0. Observe that each edge participates in a MWPM and hence $E[w]$ consists of all edges. The cycle consisting of the four orange edges, say $C$, has positive circulation  even though it is contained in $E[w]$. Cycle $C$ crosses tight odd set $S$.  

\begin{definition}
\label{def:mismatch}
	Assume that even cycle $C$ \emph{crosses} tight odd set $S$. Number the edges of $C$ starting from an arbitrary edge. Let $n_o$ and $n_e$ denote the number of odd-numbered and even-numbered edges, respectively, that cross $S$. Then the \emph{mismatch of C and S}, denoted $\mis(C, S)$, is $\abs{n_o - n_e}$.
\end{definition}

Note that in \cref{fig:blocked}, $\mis(C, S) = 2$.
Observe that if the MWPM is not unique and $C$ is a cycle in the symmetric difference of two such perfect matchings then the following must hold:
\begin{itemize}
	\item $\crc_w(C) = 0$.
	\item If $C$ crosses a tight odd set $S$, then $\mis(C, S) = 0$; the reason is that each perfect matching crosses each tight set exactly once.
\end{itemize}

\begin{figure}
	\begin{Columns}
		\Column
			\Tikz*{\pic{graph};}
		\Column
			\Tikz*{\pic{contracted};}
	\end{Columns}
	\begin{Columns}[Top]
		\Column
			\caption{The orange even cycle crosses tight odd set $S$; example due to \cite{FGT16,ST17}.}
			\label{fig:blocked}
		\Column	
			\caption{Resulting graph after shrinking tight odd set $S$.}
			\label{fig:contracted}
	\end{Columns}
\end{figure}

\begin{fact}
\label{fact:trap}
(\Cref{lem:walk-lose-edge})
		Consider a general graph $G$ with weights $w$ and even cycle $C \subseteq E[w]$ with $\crc_w(C) > 0$. Let $S$ be a tight odd set such that $C$ crosses $S$. Then $\mis(C, S) > 0$ and at least one edge of $C$ has both its endpoints in $S$. 
\end{fact}

Our strategy for dealing with cycle $C$ having $\crc_w(C)>0$ is to \emph{shrink} the tight odd set $S$ it crosses; this is illustrated in \cref{fig:contracted}. By \cref{fact:trap}, this will shrink at least one edge of $C$, hence resulting in a smaller graph. Our overall strategy is as follow: Suppose w.r.t. weight vector $w$, $\crc_w(C) = 0$. Let $w'$ be a weight vector such that $\crc_{w'}(C) > 0$. If so, \textcite{ST17} show that either $C$ must lose an edge in going from $E[w]$ to $E[w']$ or a new odd set $S$ goes tight w.r.t.\ $w'$ such that $C$ crosses $S$. In the latter case, we shrink $S$. In either case we will obtain a smaller graph and in both cases we will say that $C$ is \emph{destroyed}.

\subsubsection{Component based on graph-theoretic facts}
As stated in the Introduction, the overall structure of our algorithm is similar to that of  \cite{AV18}. Both algorithms require in each iteration a large enough number of edge-disjoint even cycles whose destruction will result in the removal of a corresponding number of edges. However, in both cases, the graph may have not such cycles. The recourse is to resort to even walks.

\begin{definition}
(\textcite{ST17})
\label{def:walk}
	We call an ordered list of an even number of edges $C = (e_1, \ldots , e_{2k})$, not necessarily distinct, that start and end at the same vertex, an \emph{even walk} if this list traverses either a simple even cycle or two odd cycles with a path joining them; in the latter case, the cycles are traversed once each and the path twice, once in each direction.
\end{definition}

The list $C$ of edges of an even walk contains each edge either once or twice, and if it contains an edge $e$ twice, then both copies will have the same parity. The notions of circulation and mismatch can be extended to even walks in a natural way by taking into consideration multiplicity of edges. Thus if $e$ occurs twice in walk $C$, is odd-numbered and crosses tight odd set $S$, then it contributes 2 to $n_o$ in the computation of $\mis(C, S)$ (see  \cref{def:mismatch}) and it contributes $2w_e$ to the sum of odd-numbered edges in the computation of $\crc_w(C)$ (see \cref{def:circulation}). As shown in \cite{ST17}, all statements made above about destroying even cycles carry over to even walks as well. 

\cite{AV18} critically used Euler's formula and the planar dual of $G$ for first finding a large number of edge-disjoint cycles in \NC. If more than half were even, they sufficed. Otherwise, they paired up odd cycles and found paths connecting each pair to obtain even walks. This was done in a such a manner that the resulting even walks were edge-disjoint.

Finding edge-disjoint cycles in a general graph in \NC{} appears to be quite difficult. Instead, we take a cue from the bipartite case, which finessed the issue of finding edge-disjoint cycles by   using \cref{fact:bipartite}. As a result, showing the \emph{existence} of cycles sufficed! However, there is a subtle difference: in the bipartite case, we needed to upper bound the number of cycles that needed to be destroyed in each iteration, whereas here we need to lower bound them; the latter is the case in \cite{AV18} as well.

Using ideas from \textcite{CPR03} we show that if the graph $G=(V, E)$ is not very sparse (see \cref{def:sparse}), then it contains $\Omega\parens*{\frac {\card{E}} {\log^2 \card{V}}}$ edge-disjoint even cycles. Then, using ideas from \cite{AV18}, we show how to pair up odd cycles to form walks. Unlike \cite{AV18}, the walks don't need to be found explicitly -- establishing existence suffices.

If in an iteration the graph is very sparse, it will not have the required number of edge-disjoint cycles. For this case, we define the notion of a \emph{triad} in \cref{def:triad}; this is a tight odd set consisting of three vertices. We show that the graph has sufficiently many disjoint triads, and a maximal independent set algorithm can find a large enough subset of these in \NC. These can be shrunk simultaneously.

\subsubsection{Component based on facts from matching theory}
Suppose that in a certain iteration our algorithm is trying weight function $w$, as per \cref{fact:bipartite}. We will need to find in \NC{} a description of face $\PM[w]$, which involves, in addition to edges $E[w]$, a maximal laminar family of tight odd sets, say $\L$. As  stated in Fact \ref{fact.allowed}, computing $E[w]$ using oracle $\O$ is straightforward. However, finding family $\L$ in \NC{} is a difficult question. The difficulty is similar to that of finding a perfect matching in a graph, i.e., the presence of a plethora of solutions. Recall the ``philosophy'' of \cite{MVV87} given in \cref{sec.history}, for dealing with this issue for perfect matching, namely  attempt to narrow down the choices to one. Clearly unlike \cite{MVV87}, randomization is not a resource we can use for this purpose. The solution involves imposing more and more restrictions on the family of tight odd sets until it becomes unique! These restrictions arise from deep structural facts from matching theory. Additional facts lead to an \NC{} algorithm for computing $\L$ with the help of $\O$. These ideas are from \textcite{CGS15} and \textcite{Sankowski18} and are given in \cref{sec:face}.

 For the ``correct'' weight function, say $w$, among the set of even walks being handled in this iteration, some will be destroyed by losing an edge and some by crossing a tight odd set. By updating the edge set to $E[w]$, we can accrue the advantage from the first set of walks. For obtaining advantage from the second set of walks, for each such walk, say $C$, we need to shrink a tight odd set, say $S$, that it crosses. A major obstacle is that our algorithm does not ``know'' \emph{any} of the walks! The way we finesse this difficulty is to shrink all outermost sets of $\L$, which are clearly disjoint, in the graph on edge set $E[w]$. 
 
 Finally, among all weight functions, we will pick the one, say $w$, that yields a graph with the smallest number of edges. There is no guarantee that $w$ would have destroyed all $s$ walks which we had established the existence of up-front. However, at least one of the weight functions must have done so and therefore led to a decrease of at least $s$ edges. Hence, $w$ must also decrease at least $s$ edges, and that suffices for making progress. As shown in \cref{lem:progress}, the number of non-isolated edges gets reduced by a factor of $1-\Omega(1/\log^2 \card{V})$ in each iteration.

\subsection{The final idea: balanced viable set}
\label{sec.balanced}

Our current strategy is to iteratively reduce the number of edges until a perfect matching remains. After picking its edges, we need to recursively find a perfect matching in each of the shrunk sets (after removing its matched vertex). The resulting algorithm would have polylogarithmic depth; however, it does not run in polylogarithmic time because of the following inherent sequentiality: Perfect matchings in shrunk sets can be found only \emph{after} finding a perfect matching in the shrunk graph, because the algorithm needs to know the vertex in $S$ that is matched outside $S$. Moreover, perfect matchings in the shrunk graph and the shrunk sets need to be found via a recursive application of the full algorithm described so far. 

The exact same issue arose in \cite{AV18} as well. The solution proposed there was meant for general graphs and hence it works here as well. The solution is quite elaborate and hence is not repeated here; instead, we direct the reader to Section 4.2 in \cite{AV18}. We note that the task is somewhat easier here because we have recourse to oracle $\O$; \cite{AV18} had to resort to computing Pfaffians orientations, etc. We give a short, high-level summary below.

An odd set $S$ is \emph{viable} if there is at least one perfect matching in $G$ which picks exactly one edge from $\delta(S)$. A set $S$ is \emph{balanced} if both $S$ and its complement contain a constant fraction of the vertices. \cite{AV18} show how to find in \NC{} a balanced viable odd set. Let $S$ be such a set. Clearly, using oracle $\O$, we can find an edge $e \in \delta(S)$  which is the unique edge in a perfect matching from this cut. Now we are done by a simple divide-and-conquer strategy: match $e$, remove its end-points and find perfect matchings in the two sides of the cut recursively, in parallel. Observe that even though perfect matchings in the two sides can be found only after finding the matched edge $e$, the latter can be done without any recursive calls, hence, leading to a polylogarithmic running time.

%%%%%%%%%%%%%%%%%%%%%%%%%%

	\section{Preliminaries}\label{sec:prelims}

We represent undirected graphs by $G=(V,E)$, where $V$ is the set of vertices and $E$ is the set of edges. Unless otherwise specified, we only work with graphs that have no loops, i.e., an edge from a vertex to itself. An edge between vertices $u$ and $v$ is represented as $\set{u, v}$. For a set $S\subseteq V$, we use $\delta(S)$ to denote the cut between $S$ and its complement, i.e., $ \delta(S)=\set{\set{u, v}\in E\given u\in S, v\notin S}$. When $S$ is a singleton, i.e., $\set{v}$ for some $v\in V$, we use the shorthand $\delta(v)=\delta(\set{v})$. A perfect matching is a subset of edges $M\subseteq E$ such that for all $v\in V$ we have $\card{M\cap \delta(v)}=1$.

\begin{definition}\label{def:isolated-edge}
	We call an edge $e=\set{u, v}$ \emph{isolated} if $\deg(u)=\deg(v)=1$.
\end{definition}
By this definition a graph is a perfect matching if it has no isolated vertices and all of its edges are isolated.

For a set $S\subseteq E$ of edges we use $\1_S\in \R^E$ to denote the indicator of $S$. We use the shorthand $\1_e$ to denote the $e$-th element of the standard basis for $\R^E$, where $e\in E$. We denote the standard inner product between vectors $w, x\in \R^E$ by $\dotprod{w, x}$.

Given a convex polytope $P\subseteq \R^E$, and a weight vector $w\in \R^E$, we use $P[w]$ to denote the set of points minimizing the weight function $x\mapsto \dotprod{w, x}$:
\[ P[w]=\set{x\in P\given \forall y\in P: \dotprod{w, x}\leq \dotprod{w, y}}. \]
Note that $P[w]$ is a face of $P$; all faces of $P$ can be obtained as $P[w]$ for appropriately chosen $w$.

\subsection{The perfect matching polytope}
Given a graph $G=(V, E)$, we call a subset of edges $M\subseteq E$ a perfect matching if it contains exactly one edge in every degree cut, i.e., $\card{M\cap \delta(v)}=1$ for all $v$. We call a graph matching-covered if any of its edges can be extended to a perfect matching.
\begin{definition}\label{def:matching-covered}
	A graph $G=(V, E)$ is \emph{matching-covered} if for every edge $e\in E$, there exists a perfect matching $M$ such that $e\in M$.
\end{definition}

The perfect matching polytope for $G=(V, E)$ is the convex hull of all perfect matchings of $G$ in $\R^E$. Thus,
\[ \PM_G=\conv\set{\1_M\given M\subseteq E\text{ is a perfect matching of }G}. \]
Clearly the perfect matchings of $G$ are in one-to-one correspondence with the vertices of this polytope.

When $G$ is clear from context, we simply use $\PM$ to refer to this polytope. $\PM$ is alternatively described by the following set of linear equalities and inequalities \cite{Edmonds65}:
\begin{equation}\label{eq:perfect-matching-polytope}
	\PM=\set*{x\in \R^E\given
		\begin{array}{ll}
			\dotprod{\1_{\delta(v)}, x}=1& \forall v\in V,\\
			\dotprod{\1_{\delta(S)}, x}\geq 1& \forall S\subseteq V,\text{ with $\card{S}$ odd},\\
			\dotprod{\1_e, x}\geq 0&\forall e\in E.
		\end{array}
	}.
\end{equation}

Any face $F$ of $\PM$ can be either described by a weight vector $w$, i.e., $F=\PM[w]$, or it can be alternatively described by the set of inequalities turned into equalities in \cref{eq:perfect-matching-polytope}. These correspond to odd sets $S$ and edges $e$. When face $F$ is clear from context, we call odd sets whose inequalities have been turned into equalities, \emph{tight} odd sets. We call an edge $e$ \emph{allowed} if $x_e>0$ for some $x\in F$, i.e., if the inequality corresponding to $e$ in \cref{eq:perfect-matching-polytope} has not been turned into equality. We use $E[w]$ or $E[F]$ to denote the set of allowed edges in the face $F=\PM[w]$. Putting it all together, to describe a face $F$ it is enough to describe the set of allowed edges as well as tight odd sets.

\subsection{Finding a description of a face}
\label{sec:face}

A key step in our oracle-based algorithm is: given small weights $w$, compute a description of the face $F= \PM[w]$. As stated before, using oracle $\O$, $E[w]$ can be computed in \NC{}. However, as far as tight odd sets go, there are typically exponentially many choices of a family of such sets that suffice. At this point, it will be useful to recall the ``philosophy'' of \cite{MVV87} given in \cref{sec.history}, namely when designing an \NC{} algorithm, faced with a plethora of solutions, one should attempt to narrow down the choices to one. Clearly unlike \cite{MVV87}, randomization is not a resource we can use for this purpose. The solution to this puzzle is indeed one of the keys that enables our result and is described below. It involves imposing more and more structure on the family of tight odd sets we seek until it becomes unique! It turns out that the latter can be computed in \NC{} with the help of $\O$.

Two tight odd sets $S_1, S_2\subseteq V$ are said to \emph{cross} if they are not disjoint and neither is a subset of the other. A family of these sets $\L\subseteq 2^V$ is said to be \emph{laminar} if no pair of sets in it cross. It is well-known that each face $F$ of the perfect matching polytope can be described by the set of allowed edges and a laminar family of tight odd sets $\L$:
\[ F = \set*{x\in \PM\given \begin{array}{ll}
	\dotprod{\1_{\delta(S)}, x}=1&\forall S\in \L,\\
	\dotprod{1_e, x}=0&\forall e\notin E[F]. 
\end{array}}. \]
In fact, $\L$ can be taken to be any \emph{maximal} laminar family of tight odd sets for the given face $F$ \cite[see, e.g.,][Lemma 2.2]{ST17}. (Note that we will always include all singletons $\set{v}$ in the laminar family $\L$ since the equalities $\dotprod{\1_{\delta(v)}, x}=1$ are automatically satisfied over all of $\PM$.) However, there are still potentially many choices for the laminar family $\L$ describing face $F$, so we impose more conditions on $\L$.

\begin{definition}\label{def:laminar-dual}
Suppose we are given a face $F=\PM[w]$. A \emph{laminar optimal dual solution} is a laminar family $\L$ of tight odd sets, including all singletons, together with a function $\pi:\L\to \R$ such that for $S \in \L$, $\pi(S)>0$ whenever $\card{S}>1$ and for all edges $e$
	\[ w_e\geq \sum_{S\in \L: e\in \delta(S)}\pi(S), \]
	with equality for allowed edges.
\end{definition}

This definition gives dual solutions for the linear program $\min\set{\dotprod{w, x}\given x\in \PM}$ that satisfy complimentary slackness and are in \emph{laminar} form. By complimentary slackness, for any such solution, $\sum_{S\in \L}\pi(S)$ is equal to the weight of a MWPM. Laminar optimal dual solutions exist but are still not unique. 

 \Textcite{CGS15} showed that extra conditions can be imposed on laminar optimal dual solution to make it unique. They studied the notion of \emph{balanaced critical dual solutions} and they showed how this unique $\L$ can be found by computing \emph{primal} solutions to the MWPM problem. \Textcite{Sankowski18} used this procedure to design an alternative \NC{} algorithm for planar graph perfect matching. We describe this procedure below. For more details see the work of \textcite{CGS15}. Note that we will not use these rather complex and elaborate extra conditions in any other context, so we will not state them explicitly. 

The following was shown by \textcite{CGS15}.
\begin{lemma}[\cite{CGS15}, Lemma 28]\label{lem:balanced}
	If $E[w]$ is connected, then a balanced critical dual is unique and \cref{alg:balanced} finds its support, the laminar family $\L$.
\end{lemma}
\begin{Algorithm}
\caption{Finding a balanced critical dual.}\label{alg:balanced}
	$\L \leftarrow \set{\set{v}\given v\in V}$.\;
	\ParallelFor{$v\in V$}{
		$\mu(v)\leftarrow \min\set{\dotprod{w, \1_M}\given M\subseteq E[w]\text{ is a perfect matching on }V\setminus\set{u, v}\text{ for some vertex }u}.$\;
	}
	Let $w'_e\leftarrow w_e+\mu(u)+\mu(v)$ for each $e\in E[w]$.\;
	\ParallelFor{$t \in \set{w'_e\given e\in E[w]}$}{
		Find the connected components of the graph $(V, \set{e\in E[w]\given w'_e\leq t})$.\;
		Add each nontrivial connected component to $\L$.\;
	}
	\Return $\L$.\;
\end{Algorithm}

It was observed by \textcite{Sankowski18} that all steps of \cref{alg:balanced} can be performed in \NC{} except for finding allowed edges $E[w]$ and the computation of $\mu(v)$'s. We note that using oracle $\O$, both these steps can be also be performed in \NC.

\begin{remark}
	When $E[w]$ is not connected, \cref{alg:balanced} still works but should be run in parallel for \emph{each connected component} of $E[w]$.
\end{remark}

\subsection{Contraction of tight odd sets, matching minors, and triads}

\textcite{Edmonds65} observed that if a collection of tight odd sets are disjoint, one can shrink each one to a single node and obtain a smaller graph whose perfect matchings can be extended to perfect matchings in the original graph. For the sake of completeness we state and prove this fact here.
\begin{fact}\label{fact:contraction}
	Suppose that $F=\PM[w]$ is a face of the matching polytope for $G=(V, E)$ and $S_1,\dots,S_k$ are tight odd sets w.r.t.\ $F$. Let $H$ be obtained from $G$ by removing disallowed edges and contracting each $S_i$ to a single node. Then any perfect matching in $H$ can be extended to a perfect matching in $G$.
\end{fact}
\begin{proof}
	Suppose that $M$ is a perfect matching in $H$. We can think of edges in $M$ as edges in $E$ as well; in fact $M\subseteq E[w]$, because we remove disallowed edges to obtain $H$. Because $M$ is a perfect matching in $H$, for each $S_i$, there is a unique $e_i\in M\cap \delta_G(S_i)$. Now since $e_i$ is an allowed edge, there must be some perfect matching $M_i$ of $G$ such that $e_i\in M_i$ and $\1_{M_i}\in F$. Since $S_i$ is a tight odd set, $M_i$ cannot have any other edge in $\delta(S_i)$, except for $e_i$. So if we look at $\set{\set{u, v}\in M_i\given u, v\in S_i}$, we must have a matching covering all vertices of $S_i$ except for the endpoint of $e_i$. Combining all of these matchings for $i=1,\dots,k$ together with $M$ will give us a perfect matching in $G$ as desired.
\end{proof}

Note that the graph $H$ obtained above is a minor of the graph $G$. But it is not an arbitrary minor. It has the additional property that every perfect matching of it can be extended back to a perfect matching of the original graph. For convenience we name these minors, matching minors.
\begin{definition}\label{def:matching-minor}
	A matching minor $H$ of a graph $G$, is a graph that can be obtained by a sequence of the following operations: Pick a face of the matching polytope and a collection of disjoint tight odd sets. Remove disallowed edges, and contract each tight odd set into a single node.
\end{definition}
The following statement follows directly from \cref{fact:contraction}.
\begin{lemma}\label{lem:matching-minor}
	If $H$ is a matching minor of the graph $G$, then every perfect matching in $H$ can be extended to a perfect matching in $G$.
\end{lemma}

In our algorithms, we use the simple observation that a path of length $2$ on vertices of degree $2$ yields a tight odd set for the entire matching polytope. We call these paths triads.

\begin{definition}
\label{def:triad}
	A \emph{triad} in graph $G=(V, E)$ is a set of three vertices $\set{a, b, c}$ such that $\deg(a)=\deg(b)=\deg(c)=2$, and $\set{a, b}, \set{b, c}\in E$. 
\end{definition}

\begin{lemma}\label{lem:triad}
	A triad $\set{a, b, c}$ is a tight odd set for the matching polytope and all of its faces.
\end{lemma}
\begin{proof}
	The only two neighbors of $b$ are $a, c$. So in every perfect matching, $b$ must be matched to one of them. The other vertex must have an edge to an outside vertex, and in fact that is the only possible edge in $\delta(\set{a, b, c})$.
\end{proof}
\begin{remark}
	Note that the proof of \cref{lem:triad} does not use the assumptions $\deg(a)=\deg(c)=2$ and only uses $\deg(b)=2$. We will use these extra assumptions elsewhere, to prove that in certain situations, we can find many triads in our graph.
\end{remark}

\subsection{Even walks and weight vectors}
\label{sec:walks}

Even walks were defined in \cref{def:walk}. For an even walk $C$, define the \emph{signature} of $C$ to be the vector:
\[ \sign(C) = \sum_{i=1}^{2k} {(-1)^i \1_{e_i}} . \]
The notions of circulation and mismatch can be stated in terms of signature: 
 \[\crc_w(C) = |\dotprod{w, \sign(C)}| \]
 \[ \mis(C, S) = |\dotprod{\1_{\delta(S)}, \sign(C)}| \]

Now, there cannot be two distinct points $x, y\in \PM[w]$ whose difference $x-y$ is a multiple of $\sign(C)$, since otherwise we would have $\dotprod{w, x}\neq \dotprod{w, y}$. Another way of stating this is that if $x\in \PM[w]$, then $x+\epsilon\sign(C)\notin \PM[w]$ for any $\epsilon\neq 0$. So, some inequality or equality describing $\PM[w]$ must be violated for this point. If we pick $x$ to be in the relative interior of the face $\PM[w]$ we will have some slack for non-tight inequalities describing $\PM[w]$. So the violated constraint for $x+\epsilon\sign(C)$ must be a constraint that is tight for the entire face $\PM[w]$. This implies that:

\begin{lemma}
\label{lem:walk-destroy}
	Let $C$ be an even walk with $\crc_w(C) > 0$. Then either there is an edge $e\in C$ that is disallowed, i.e., $e\notin E[w]$, or for any laminar dual $(\L, \pi)$ describing $\PM[w]$, there is some set $S$ such that $\mis(C, S) > 0$.
\end{lemma}

For a more detailed proof of this, see \cite{AV18}. Note that if $\mis(C, S) > 0$, then $C$ must have an edge with both endpoints inside $S$.

\begin{lemma}
\label{lem:walk-lose-edge}
	Suppose that $C$ is an even walk and $S$ is a tight odd set such that $\mis(C, S) > 0$. Then there is an edge $e=\set{u, v}\in C$ such that $u, v\in S$.
\end{lemma}

\begin{proof}
	If this is not true, then every time $C$ enters $S$ it must immediately exit. So if we compute $\mis(C, S)$ by looking at edges that cross $S$, we always get a $+1$ followed by a $-1$, and a $-1$ followed by a $+1$. So the entire sum would be $0$ which is a contradiction.
\end{proof}

We also borrow from \textcite{FGT16} the following important result, which is also stated in \textcite{ST17} and as \cref{fact:bipartite} in this paper.
\begin{lemma}[\cite{FGT16}]\label{lem:w-exists}
	There is a polynomial sized family of polynomially bounded weight vectors $\W$, such that for any set of edge disjoint even walks $C_1,\dots,C_k$, there is some $w\in \W$ which ensures
	\[ \forall i: \crc_w(C) > 0. \]
\end{lemma}
\begin{proof}
	This lemma is actually proved in \cite{FGT16, ST17} for any collection of nonzero vectors, not just $\sign(C_i)$'s, as long as there is both a polynomial bound on the number of vectors and the absolute value of their coordinates. Edge-disjointness of even walks automatically puts a bound of $\card{E}$ on their number, and the coordinates of our even walks are always bounded in absolute value by $2$.
\end{proof}

\subsection{Maximal independent sets}

Given a graph $G=(V, E)$, we call a subset $S\subseteq V$ independent if no edge $e\in E$ has both endpoints in $S$. We call an independent set maximal if no strict superset $T\supsetneq S$ is independent. We will crucially use the fact that maximal independent sets can be found in \NC{}.
\begin{theorem}[\cite{Luby86}]\label{thm:luby}
	There is a deterministic \NC{} algorithm that on input graph $G=(V, E)$ returns a maximal independent set $S\subseteq V$.
\end{theorem}
We usually want a large, rather than a maximal, independent set. We will use the fact that in bounded degree graphs, any maximal independent set is automatically large.
\begin{fact}\label{fact:bounded-degree-ind}
	If $G=(V, E)$ is a graph with $\deg(v)\leq \Delta$ for all $v\in V$, then any maximal independent set $S\subseteq V$ satisfies
	\[ \card{S}\geq \frac{\card{V}}{\Delta+1}. \]
\end{fact}

	\section{The Decision Oracle}\label{sec:oracle}

We will assume that our algorithm is equipped with an oracle $\O$ which answers the following type of queries: Given a graph $G=(V, E)$ and small weights $w\in \Z^E$, what is the weight of a MWPM in $G$? We denote the answer by
\[\O(G, w)=\min\set{\dotprod{w, x}\given x\in \PM_G}.\]
	
We now list several deterministic \NC{} primitives based on $\O$. Versions of these two lemmas appear implicitly, stated for planar graphs, in \textcite{Sankowski18}, but we prove them for the sake of completeness.
\begin{lemma}\label{lem:support}
	Given access to $\O$, for polynomially bounded $w\in \Z^E$, one can find $E[w]$ in \NC.
\end{lemma}
\begin{proof}
	An edge $e=\set{u, v}$ can be in a MWPM if and only if $\O(G, w)=w_e+\O(G-\set{u}-\set{v}, w)$, where $G-\set{u}-\set{v}$ is obtained from $G$ by removing vertices $u, v$. This can be checked in parallel for all edges $e$.
\end{proof}
\begin{lemma}
\label{lem:laminar}
	Given access to $\O$, for polynomially bounded $w\in \Z^E$, one can run \cref{alg:balanced} in \NC.
\end{lemma}
\begin{proof}
	As was observed by \textcite{Sankowski18}, all steps of \cref{alg:balanced} can be run in \NC{} except for finding $E[w]$ and computing $\mu(v)$. Given access to $\O$, we can find $E[w]$ in \NC{} by \cref{lem:support}. Furthermore observe that for any $v\in V$
	\[ \mu(v)=\min\set{\O(G-\set{u}-\set{v}, w)\given u\in V-\set{v}}, \]
	which can be computed by making all queries $\O(G-\set{u}-\set{v}, w)$ in parallel and then taking the minimum.
\end{proof}

An implementation for the oracle, in \RNC{} with arbitrarily small inverse polynomial probability of error for general graphs, follows from \textcite{MVV87}, since they give an \RNC{} algorithm for finding a MWPM for small weights. Since $\O$ is promised to be called at most polynomially many times, the probability of error over the entire run of the algorithm can be made inverse polynomially small.

	\section{Structural Facts}
\label{sec:strucrture}

Our algorithm requires two structural facts, one for the case that the graph $G$ is very sparse and the other for the complementary case. They are encapsulated in \cref{lem:many-triads,lem:many-walks}.

\begin{definition}
\label{def:sparse}
	A connected graph $G = (V, E)$ is said to be \emph{very sparse} if $\card{E}<\card{V}/(1-\epsilon)$, for some constant $\epsilon<1/9$.
\end{definition}

\begin{lemma}\label{lem:many-triads}
	If $G=(V, E)$ is a matching-covered, very sparse graph, then the number of triads in any maximal set of node-disjoint triads in $G$ is at least $c_1\card{E}$, for some constant $c_1(\epsilon)>0$.
\end{lemma}

The proof of this lemma involves two steps: first, we prove that the total number of triads is large and second, that a maximal node-disjoint set of triads must also be large. The first step is accomplished in the following lemma.

\begin{lemma}\label{lem:path22}
	Suppose that $G=(V, E)$ is a graph with no vertices of degree $0$ or $1$. Then the number of triads in $G$ is at least $9\card{V}-8\card{E}$.
\end{lemma}

\begin{proof}
	Consider a charging scheme, where we allocate a budget of $1$ to each edge, and the edge distributes its budget between its two endpoints. We then sum up the charge on all vertices and use the fact that this sum is exactly $\card{E}$.
	
	Let $e=\set{u, v}$ be an edge. If neither $u$ nor $v$ is of degree $2$, let the edge give $1/2$ to $u$, and $1/2$ to $v$. If both $u$ and $v$ are of degree $2$, we allocate the budget the same way by splitting it equally between $u$ and $v$. The only remaining case is when one of $u$ and $v$ has degree $2$ and the other has degree at least $3$; by symmetry let us assume that $\deg(u)=2$ and $\deg(v)\geq 3$. Then we allocate $5/8$ to $u$ and $3/8$ to $v$.
	
	Now let us lower bound the charge that each vertex $v$ receives. Note that the minimum amount $v$ receives from any of its adjacent edges is $3/8$, so an obvious lower bound is $3\deg(v)/8$. If $\deg(v)\geq 3$, this is at least $9/8$. Now consider the case when $\deg(v)=2$. Then $v$ receives at least $1/2$ from each of its adjacent edges. If one of the neighbors of $v$ is not of degree $2$, then the charge that $v$ receives will be at least $1/2+5/8=9/8$. The only possible case where $v$ does not receive at least $9/8$ is when it is of degree $2$, and both of its neighbors are also of degree $2$ (the center of a triad), in which case it receives $1$.
	
	Now let $k$ be the number of triads. Then, by the above argument the total charge on all the vertices is at least
	\[ \frac{9}{8}(\card{V}-k)+k\leq \card{E}. \]
	Rearranging yields $k\geq 9\card{V}-8\card{E}$.
\end{proof}

\begin{proof}[Proof of \cref{lem:many-triads}]
	We know that the number of triads is at least $9\card{V}-8\card{E}=(1-9\epsilon)\card{E}$. Now consider the conflict graph of triads, where nodes represent triads, and edges represent having an intersection. It is easy to see that any triad can only intersect at most $4$ other triads. So the degrees in this conflict graph are bounded by $4$. By \cref{fact:bounded-degree-ind}, any maximal node-disjoint set of triads will contain at least $(1-9\epsilon)\card{E}/5$ many triads. So we can take $c_1(\epsilon)=(1-9\epsilon)/5$ which is positive for $\epsilon<1/9$.
\end{proof}

\begin{lemma}\label{lem:many-walks}
	If $G=(V, E)$ is a matching-covered graph on $\card{V}>2$ vertices that is not very sparse, then there exist $c_2\card{E}/\log^2 \card{V}$ edge-disjoint even walks in $G$, for some constant $c_2(\epsilon)>0$.
\end{lemma}

We first show that there are many edge-disjoint cycles in a non-sparse graph. If at least half of them are even, we are done. Otherwise, we show how to pair up odd cycles and connect them via suitable paths to get sufficiently many edge-disjoint even walks. A proof of the next lemma can be found in \cite{CPR03}; however, for the sake of completeness we provide it here.

\begin{lemma}\label{lem:cycles}
	In a graph $G=(V, E)$ there exists a collection of edge-disjoint cycles with at least the following number of cycles:
	\[ \frac{\card{E}-\card{V}}{2\log_2 \card{V}}. \]
\end{lemma}

\begin{proof}
	We prove this by induction on $\card{V}+\card{E}$. We have several cases:
	\begin{enumerate}[i)]
		\item If there are any loops in the graph, we extract that as one of our cycles, and remove the edge from the graph. The promised quantity goes down by $1/(2\log_2 \card{V})$ which is $\leq 1/2$. So from now on we assume that $G$ has no loops.
		\item If there are any two parallel edges $e, e'$, we extract those as a cycle of length $2$, and remove both from the graph. The promised number of edge-disjoint cycles goes down by $2/(2\log_2 \card{V})\leq 1$. So adding the cycle we extracted fulfills the promise. From now on we assume that $G$ is simple.
		\item If $G$ has any vertices of degree $0$: We can simply remove it and the promised quantity grows.
		\item If $G$ has a vertex of degree $1$: We can also remove this vertex. This operation does not change the numerator but shrinks the denominator, which results in a larger promised quantity.
		\item If $G$ has a vertex $v$ of degree $2$: Let $e, e'$ be the two adjacent edges to $v$. Remove $v, e, e'$ from the graph, and place a new edge $e''$ between the two former neighbors of $v$. By doing this, both $\card{V}$ and $\card{E}$ go down by $1$. So now the promised number of edge-disjoint cycles becomes larger. By induction we find them, and now we replace the edge $e''$ if it is used at all in a cycle, by the path of length two consisting of $e, e'$. Since $e''$ appears in at most one cycle, this operation preserves edge-disjointness.
		\item Finally if $G$ is a simple graph with no vertices of degree $\leq 2$, it must have a cycle of length at most $2\log_2\card{V}$. If we prove this, we are done by induction, because we can remove the edges of this cycle and the promised quantity goes down by at most $1$. Now to prove the existence of this cycle, assume the contrary, that the length of the minimum cycle of the graph is at least $2\log_2\card{V}+1$. Pick a vertex $v$ and look at all simple paths of length at most $\log_2\card{V}$ going out of $v$. The number of paths of length $i$ is at least twice the number of paths of length $i-1$. This is because every path of length $i-1$ ending at a vertex $u$ can be extended in at least $\deg(u)-1\geq 2$ ways, and none of these extensions will intersect themselves, otherwise we would get a cycle of length $\log_2\card{V}+1$. So in the end, the total number of such paths will be $>2^{\log_2 \card{V}}=\card{V}$, which means that two of the paths must share an endpoint. But now from the union of these two paths, we can extract a cycle of length at most $\log_2\card{V}+\log_2\card{V}=2\log_2\card{V}$.
	\end{enumerate}
\end{proof}

If at least half of the cycles guaranteed by \cref{lem:cycles} are odd, we need to pair them up and connect them with paths. We use a spanning tree to do this.

\begin{fact}[{\cite[For proof see Lemma 20 in][]{AV18}}]\label{fact:tree-paths}
	Consider a tree $T$ with an even number of tokens placed on its vertices, with possibly multiple tokens on each vertex. There is a pairing, i.e., a partitioning of tokens into partitions of size two, such that the unique tree paths connecting each pair are all edge-disjoint.
\end{fact}

\begin{lemma}\label{lem:tree-join}
	Suppose that there are $2\l$ edge-disjoint cycles of odd length in a matching-covered connected graph $G=(V, E)$. Then $G$ contains at least $\Omega(\l^2/\card{E})$ edge-disjoint even walks.
\end{lemma}

\begin{proof}
	We will pair up the odd cycles by paths connecting each pair. This will create $\l$ even walks, but they might not be edge-disjoint. We will then show how to extract $\Omega(\l^2/\card{E})$ \emph{edge-disjoint} even walks out of them.
	
	Consider a spanning tree $T$ of $G$. For each of the $2\l$ odd cycles, pick an arbitrary vertex, and put a token on that vertex. Now we have an even number of tokens on the vertices. We can pair up these tokens, so that the unique tree paths (of possibly length $0$) connecting each pair are edge-disjoint, see \cref{lem:tree-join}.
	
	Now for each pair of odd cycles $C_1, C_2$ whose tokens got paired up, we create an even walk. Let $P$ be the tree path connecting tokens from $C_1$ and $C_2$. If $P$ has no common edges with $C_1, C_2$ we can simply create our even walk, but this is not guaranteed to happen. So instead, traverse $P$ from $C_1$'s token to $C_2$'s token and look at the last exit from $C_1$; afterwards look for the first time any vertex of $C_2$ is visited. This portion of $P$ is a subpath connecting $C_1$ and $C_2$ having no common edged with either. We use $C_1, C_2$ and this subpath of $P$ to create our even walk.
	
	So far we have created $\l$ even walks, but they might not be edge-disjoint. The odd cycles are edge-disjoint, as are the paths connecting them, but one of the paths might share an edge with an unrelated odd cycle. This also means that no edge $e$ can be shared between more than two even walks; $e$ can be used once as part of an odd cycle, and once as part of a path.
	
	Now consider the number of edges in each even walk. If we sum this over all even walks, we get at most $2\card{E}$, since each edge can appear in at most two even walks. So the average number of edges in an even walk is $\leq 2\card{E}/\l$. By Markov's inequality at least half of the even walks, $\l/2$ of them, will have at most twice this average number of edges, $4\card{E}/\l$. Now create a conflict graph where nodes represent these $\l/2$ even walks, and an edge is placed when the two even walks share an edge. The degree of each node is at most $4\card{E}/\l$. So if pick a maximal independent set in this conflict graph, it will consist of at least $\Omega(\l^2/\card{E})$ many even walks.
\end{proof}

We are finally ready to prove \cref{lem:many-walks}.

\begin{proof}[Proof of \cref{lem:many-walks}]
	First note that if our graph is not an isolated edge and is matching-covered it must contain at least one even cycle. This is so because there must be at least two perfect matchings in the graph, and in their symmetric difference, we can find one such cycle.
	
	Because we are guaranteed to have at least $1$ cycle, we can simply show that asymptotically we can extract $\Omega(\card{E}/\log^2\card{V})$ edge-disjoint even walks. Then the asymptotic statement translates to the more concrete bound of $c_2\card{E}/\log^2\card{V}$.
	
	If $(1-\epsilon)\card{E}\geq \card{V}$, by \cref{lem:cycles}, we have  \[\epsilon \card{E}/2\log_2\card{V}=\Omega(\card{E}/\log\card{V})\] cycles. If at least half of them are of even length, we are done. Otherwise we get $\Omega(\card{E}/\log \card{V})$ odd cycles. Perhaps by throwing away one of them, we can assume the number of odd cycles we have is even. Then we can apply \cref{lem:tree-join} to obtain $\Omega(\card{E}/\log^2\card{V})$  edge-disjoint walks. This completes the proof.
\end{proof}

	\section{The Oracle-Based Algorithm}\label{sec:algorithm}

In this section we describe our oracle-based algorithm for finding a perfect matching. In \cref{sec:weighted}, we will extend this to finding a \emph{minimum weight} perfect matching for small weights.

On input $G=(V, E)$, our algorithm proceeds by finding smaller and smaller matching minors $H$ of $G$, until $H$ has a unique perfect matching, or in other words is a perfect matching. Then we pick the edges in $H$ as a partial matching in $G$ and extend this partial matching to a perfect matching independently and in parallel for the preimage of each node in $H$. That is for each node $s$ in $H$, we take the set $S\subseteq V$ that got shrunk to $s$, remove the single endpoint of the partial matching from $S$, and recursively find a perfect matching in $S$. In the end we return the results of all these recursive calls along with the edges of $H$ as the final answer.

We crucially make sure that the pre-image of nodes in $H$ never contain more than a constant fraction of $V$. This makes sure that our recursive calls end in $O(\log \card{V})$ steps.

In all of our algorithms, when we construct matching minors, we implicitly maintain the mapping from the resulting edges to the original edges, and the mapping from original vertices to the minor's vertices. These are trivial to maintain in \NC{}, but for clarity we avoid explicitly mentioning them. We also keep node weights for matching minors, where the \emph{weight of a node} is simply the number of original vertices that got shrunk to it.

\begin{Algorithm}
	\caption{Divide-and-conquer algorithm for finding a perfect matching.}
	\label{alg:perfect-matching}
	\PerfectMatching{$G=(V, E)$}\;
	\eIf{$V=\emptyset$}{
		\Return $\emptyset$.\;
	}{
		Call $\PartialMatching{G}$, and let $H$ be the matching minor returned.\;
		Let $M\subseteq E$ be the edges of $H$.\;
		\ParallelFor{each node $s$ of $H$}{
			Let $S\subseteq V$ be the nodes of $G$ that are shrunk to $s$.\;
			Let $v$ be the unique endpoint of the unique edge of $M$ in $\delta(S)$.\;
			Let $G_s$ be the induced graph on $S-\set{v}$.\;
			$M\leftarrow M\cup \PerfectMatching(G_s)$.\;
		}
		\Return $M$.\;
	}
\end{Algorithm}

\begin{Algorithm}
	\caption{Find a matching minor of the input graph that is itself a perfect matching.}
	\label{alg:partial-matching}
	\PartialMatching{$G=(V, E)$}\;
	Assign node weight $1$ to each node $v\in V$.\;
	\While{$G$ is not a perfect matching}{
		\If{any node $v$ of $G$ has at least $1/6$ of the total node weight}{
			Remove disallowed edges $e\notin E[0]$ from $G$.\;
			Contract the complement of $\set{v}$ to a single node. If there are parallel edges, remove all except for an arbitrary one.\;
			\Return $G$.\;
		}
		Find a maximal set of node-disjoint triads in $G$.\;
		Let $H$ be obtained from $G$ by removing disallowed edges and contracting each triad into a single node.\;
		$U\leftarrow \set{H}$.\;
		\ParallelFor{$w\in \W$}{
			Call $\Reduce(G, w)$ and let the result be $H$.\;
			$U\leftarrow U\cup\set{H}$.\;
		}
		Find the graph $H\in U$ with the minimum number of non-isolated edges.\;
		$G\leftarrow H$.\;
	}
	\Return $G$.\;
\end{Algorithm}

\begin{Algorithm}
	\caption{Remove disallowed edges and contract certain tight odd sets.}
	\label{alg:reduce}
	\Reduce{$G=(V, E)$, $w$} \tcp*{The graph $G$ has node weights.}
	Remove disallowed edges $e\notin E[w]$ from $G$.\;
	Find all connected components of $G$.\;
	\ParallelFor{each connected component $C$ of $G$}{
		Run \cref{alg:balanced} on $C$ to find a laminar family of tight odd sets $\L$.\;
		\ParallelFor{$S\in \L$}{
			\If{node weight of $S$ is more than half of the node weight of $C$}{
				Replace $S$ in $\L$ with $C-S$.\;
			}
		}
		Find the inclusion-wise maximal sets in $\L$ and shrink each one to a single node.\;
	}
	\Return $G$.\;
\end{Algorithm}

The pseudocode for the main algorithm \PerfectMatching can be seen in \cref{alg:perfect-matching}. On input $G$, the algorithm calls \PartialMatching to find a matching minor $H$ of $G$ which  itself is a perfect matching. Then the edges of $H$, which form a partial matching in $G$, are extended to a perfect matching independently and in parallel in the preimage of each node from $H$. Since $H$ is a matching minor, this extension can always be performed by \cref{lem:matching-minor}.

The pseudocode for \PartialMatching can be seen in \cref{alg:partial-matching}. This algorithm keeps a node-weighted matching minor of the input graph $G$. It tries several ways of obtaining a smaller matching minor, where \emph{size of a matching minor} is measured in terms of the number of non-isolated edges, see \cref{def:isolated-edge}. One way of obtaining a smaller matching minor is by picking a maximal node-disjoint set of triads and shrinking them simultaneously. By \cref{lem:triad}, this produces a matching minor. Also note that the maximal set of node-disjoint triads can be found in \NC{} by enumerating all triads and using \cref{thm:luby}.

Another way of obtaining smaller matching minors is by trying weights from the set of weight vectors $\W$ and calling \Reduce to remove disallowed edges $e\notin E[w]$ and shrinking top-level sets of a laminar family of tight odd sets w.r.t.\ $w$.

Finally. the pseudocode for \Reduce can be seen in \cref{alg:reduce}. This algorithm is simply fed a graph $G=(V, E)$ and a weight vector $w$. It removes disallowed edges $e\notin E[w]$ and shrinks the maximal sets of a laminar family of tight odd set. The laminar family is found using \cref{alg:balanced}, but is modified to make sure that no shrunk set becomes too large; to be more precise no shrunk vertex in the end will have node weight more than half of the total node weight.

\subsection{Finding a minimum weight perfect matching}\label{sec:weighted}

We extend our algorithm so it returns not just any perfect matching, but rather a minimum weight perfect matching, for small weights.

Given an input graph $G=(V, E)$ and a weight vector $w$, we can remove disallowed edges $e\notin E[w]$, and find a laminar family of tight odd sets $\L$ w.r.t.\ $w$, by calling \cref{alg:balanced} on each connected component of $G$. By complementary slackness, any perfect matching that has only one edge in $\delta(S)$ for each $S\in \L$ will automatically be of minimum weight, see \cref{def:laminar-dual}. We can simply contract the top level sets in $\L$, use \cref{alg:perfect-matching} to find a perfect matching in the shrunk graph, and recursively extend this to a minimum weight perfect matching in each shrunk piece. Following an almost identical argument as in the proof of \cref{fact:contraction}, the perfect matching in the shrunk graph can be extended to a minimum weight perfect matching.

The only problem with this method is that the recursion depth is not guaranteed to be polylogarithmic. However we can fix that by making sure that tight odd sets $S\in \L$ do not have more than half of the vertices in the graph; if they do, we replace them by their complements and we will see in \cref{lem:flip} why this operation preserves laminarity.

\subsection{Minor-closed families of graphs}\label{sec:minor-closed}

Throughout our algorithm we only call the decision oracle on graphs obtained from the original through a sequence of edge and vertex removals and contractions. In this section we will prove that the decision oracle is only called on minors of the original graph, that is those graphs obtained by vertex and edge removals and contractions of \emph{connected} subgraphs.

\begin{lemma}\label{lem:minor}
	\Cref{alg:reduce,alg:partial-matching,alg:perfect-matching} call the decision oracle on minors of their input graph only.
\end{lemma}

This lemma is all we need to prove \cref{thm:minor}. Note that there are several minor-closed families of graphs where the decision problem can be solved in \NC{} by using a counting oracle. In particular we can count perfect matchings in graphs embedded on surfaces of genus at most $O(\log n)$, and therefore solve the decision problem, all in \NC{}. This improves upon the genus bound of $O(\sqrt{\log n})$ given by \textcite{AV18}.

\begin{corollary}
	For graphs embedded on a surface of genus at most $O(\log n)$ and weighted with polynomially bounded edge weights, there is an \NC{} algorithm to find a minimum weight perfect matching.
\end{corollary}

Another consequence of \cref{thm:minor} is an alternative algorithm for $K_{3,3}$-free graphs, which was resolved earlier by \textcite{EV18}.

Now we prove \cref{lem:minor}.

\begin{proof}[Proof of \cref{lem:minor}]
	First we prove this for \cref{alg:reduce}. In this algorithm, we only remove edges from the input graph, and shrink tight odd sets in connected components. We just have to show that what we shrink is already connected. Consider a tight odd set $S$ in a connected component $C$. If it is not internally connected, then one of its internal connected components must have odd size; let that be $S'$. Since $S\subseteq C$ and $C$ is a connected component, there is an edge $e\in \delta(S-S')$. Since $S'$ is not internally connected to $S-S'$, it must be that $e\in \delta(S)$ too. Now since the graph is matching-covered with minimum weight perfect matchings, there must be some minimum weight perfect matching $M\ni e$. But because $S'$ is odd, there must also be an edge $f\in M\cap \delta(S')$. But note that $e\neq f$, and both $e,f\in \delta(S)$. This is a contradiction, since $S$ cannot have more than one edge in a perfect matching. This shows that $S$ must be connected and \cref{alg:reduce} only produces minors of its input graph.
	
	Next we prove the statement for \cref{alg:partial-matching}. This algorithm either calls \cref{alg:reduce}, or finds triads and contracts them. The former produces minors of the input graph, and the latter also produces minors of the input graph since triads are connected.
	
	Note that the graph returned by \cref{alg:partial-matching} may not be a proper minor of the input graph; that could happen if the node weight of some $v$ goes above $1/6$ the total node weight. In this scenario, the complement of $v$ might not be connected and yet we contract it. However the algorithm immediately returns and the decision oracle is not called on this returned graph. So this does not contradict the statement of the lemma.
	
	Finally we prove the statement for \cref{alg:perfect-matching}. The only graphs produced and passed onto \cref{alg:partial-matching} are obtained from the input graph by vertex removals and edge removals. So they are all minors of the input graph. The output of \cref{alg:partial-matching} might not be a proper minor, but this output is only used to decide which edges and vertices to remove from the original graph to get to induced graphs on $S-\set{v}$.
\end{proof}
	\section{Analysis of the algorithm}\label{sec:analysis}

First we will prove that our oracle-based algorithm returns a correct answer. Next, we will bound the running time and prove that our algorithm runs in \NC, modulo the calls to $\O$; this constitutes the most challenging part of the analysis.

\subsection{Correctness}

We will need the following lemma.

\begin{lemma}\label{lem:flip}
	Suppose that $\L$ is a laminar family of sets in a node-weighted graph $G=(V, E)$, and we replace every $S\in \L$ whose node weight is larger than half of the total node weight by the complement, i.e., $V-S$. Then the resulting family of sets $\L'$ is also laminar.
\end{lemma}
\begin{proof}
	Let $S, S'$ be two sets in $\L$. They are either disjoint or one is contained in the other.
	
	If $S\cap S'=\emptyset$: They cannot both have node weight more than $1/2$. So at most one of them gets replaced by its complement. Then it is easy to see that the resulting sets do not cross.
	
	If $S\subseteq S'$: There are three possibilities. If none of them gets replaced by their complements, or both of them get replaced by their complements, they remain nested and therefore do not cross. If one of them gets replaced by its complement, it has to be the larger set $S'$. In that case the resulting sets become disjoint, and still do not cross.
\end{proof}

Using \cref{lem:flip} and \cref{lem:matching-minor}, we deduce that \Reduce always returns a matching minor of its input graph. By definition, \PartialMatching also returns a matching minor of its graph when it finishes (for the analysis of running time see \cref{sec:runtime}).

This proves the correctness of the algorithm, since we always find a matching minor that has a unique perfect matching (itself), and by \cref{lem:matching-minor}, we can extend it to a perfect matching, independently in the preimage of each node.

\subsection{Running time}\label{sec:runtime}

First we analyze \PerfectMatching (\cref{alg:perfect-matching}) assuming the calls to \PartialMatching (\cref{alg:partial-matching}) are in \NC{}.

\begin{lemma}
	Assuming the calls to \PartialMatching are in \NC{}, then the procedure \PerfectMatching is in \NC{}.
\end{lemma}
\begin{proof}
	We simply need to bound the number of levels in the recursion. To do so, we will prove that when \PartialMatching returns a matching minor $H$, the node weight of every node is at most $5/6$ the total node weight. This proves that in each recursive call to \PerfectMatching, the number of vertices gets reduced by a factor of $5/6$.
	
	Note that the first time in \cref{alg:partial-matching} that a node's weight goes above $1/6$ the total weight, the algorithm stops and returns a two-node minor. So we just need to prove that the weight of the node that just went above $1/6$ is not more than $5/6$. The current minor was obtained from the previous minor by either \Reduce, or by shrinking triads. But \Reduce never creates nodes with weight more than half the total weight. The weight of each node in a triad is also at most $1/6$ the total weight, so after shrinking the triad, the new weight can be at most $1/6+1/6+1/6=1/2$ the total weight. This finishes the proof.
\end{proof}

Finally, we need to prove that \PartialMatching finishes in a polylogarithmic number of steps. Using the structural facts, \cref{lem:progress} and \cref{lem:many-triads,lem:many-walks}, we establish the following lemma.

\begin{lemma}
\label{lem:progress}
	In each iteration of \cref{alg:partial-matching}, the number of non-isolated edges gets reduced by a factor of $1-\Omega(1/\log^2 \card{V})$.
\end{lemma}

\begin{proof}
	First assume $G$ is a connected graph. Then we can directly apply \cref{lem:many-triads,lem:many-walks} for some fixed $\epsilon<1/9$ to show that we either find $c_1\card{E}$ triads or there exist $c_2\card{E}/\log^2\card{V}$ edge-disjoint even walks. In the former case, after contracting the triads, the number of edges gets reduced by a factor of $1-c_1$. In the latter case, let $C_1,C_2,\dots,C_k$ be the edge-disjoint even walks, and let $w\in \W$ be the weight vector such that $\dotprod{w, \sign(C_i)}\neq 0$. Note that $w$ is guaranteed to exist by \cref{lem:w-exists}. In the call to $\Reduce(G, w)$, every $C_i$ loses at least edge by \cref{lem:walk-destroy,lem:walk-lose-edge}, either because one of its edges becomes disallowed or it gets shrunk as a result of shrinking top-level tight odd sets. Therefore, one of the candidate graphs in $U$ in \cref{alg:partial-matching} will have a factor of $1-c_3/\log^2\card{V}$ fewer edges, for some constant $c_3>0$.
	
	Next assume $G$ is not connected. If so, we apply the above-stated argument to each connected component that is not an isolated edge. We can further assume the same weight vector $w$ works for all connected components. Now if $H_1$ is the graph obtained from shrinking triads, and $H_2$ is the result of $\Reduce(G, w)$, then we know that the average number of edges in $H_1$ and $H_2$ for each connected component is at most $1-c_3/2\log^2 \card{V}$ times the number of edges in the connected component. So one of $H_1, H_2$ must have at most $(1-c_3/2\log^2\card{V})$ times as many non-isolated edges as $G$.
\end{proof}

Note that \cref{lem:progress} gives a polylogarithmic upper bound on the number of iterations in \cref{alg:partial-matching}, since if we track the number of non-isolated edges, after every $\Theta(\log^2\card{V})$ steps we get a constant factor reduction, and therefore it takes at most $O(\log \card{E}\cdot \log^2\card{V})$ iterations for it to reach $0$.

	\section{Discussion}\label{sec:discussion}

This paper has identified what appears to be the ``core'' of the difficult open problem of obtaining an \NC{} matching algorithm, namely the decision problem. We must immediately mention that both decision problems stated in \cref{sec:right} have been the subject of numerous attacks over the past decades and hence resolution is not likely to be an easy matter. At the same time, we hope that since the ``target'' has been more precisely identified, the resolution of the open problem will gain added impetus. 

An obvious open question is to build on the \QuasiNC{} algorithms of \textcite{GT17, GTV17} to obtain the appropriate oracle-based \NC{} algorithms and pseudo-deterministic \RNC{} algorithms for linear matroid intersection and for finding a vertex of a polytope with faces given by totally unimodular constraints. An interesting problem defined by Papadimitriou and Yannakakis \cite{PY} , called Exact Matching, is the following: Given a graph $G$ with a subset of the edges marked red and an integer $k$, find a perfect matching with exactly $k$ red edges. This problem is known to be in \RNC~\cite{MVV87}, even though it is not yet known to be in \P. Is there a pseudo-deterministic \RNC{} algorithm for it?

The phenomenon identified in \cref{sec:curious} clearly deserves to be studied in depth. To the best of our knowledge, there are only two algorithmic results for bipartite matching that have not been extended to general graphs. The first is obtaining a fully polynomial randomized approximation scheme for counting the number of perfect matchings \cite{JSV04}; this is also among the outstanding open problems of theoretical computer science today. The second is obtaining an $O(m^{11/8})$ algorithm for maximum matching \cite{Sidford-L}, which beats the earlier algorithms for sparse graphs.

	\PrintBibliography
\end{document}